\documentclass[12pt]{article}
\usepackage[utf8]{inputenc}
\usepackage{fullpage, amsmath, amsfonts, amsthm, color, graphicx}
\usepackage{algorithm}
\usepackage{algpseudocode}
\usepackage{hyperref}

\usepackage{authblk}

\newtheorem{theorem}{Theorem}
\newtheorem{lemma}{Lemma}
\newtheorem{definition}{Definition}

\newcommand{\sgn}{\text{sgn}}
\renewcommand{\Re}{\text{Re}}

\title{Optimizing quantum circuit parameters via SDP}
\author{Eunou Lee\thanks{\noindent QCenter, Sunkyunkwan University, South Korea \indent 
\href{mailto:eunoulee@skku.edu}{eunoulee@skku.edu}
}}

\date{}

\begin{document}

\maketitle
\begin{abstract}
In recent years, parameterized quantum circuits have become a major tool to design quantum algorithms for optimization problems.
The challenge in fully taking advantage of a given family of parameterized circuits lies in finding a good set of parameters in a non-convex landscape that can grow exponentially to the number of parameters. 

We introduce a new framework for optimizing parameterized quantum circuits: round SDP solutions to circuit parameters.
Within this framework, we propose an algorithm that produces approximate solutions for a quantum optimization problem called Quantum Max Cut.
The rounding algorithm runs in polynomial time to the number of parameters regardless of the underlying interaction graph.

The resulting 0.562-approximation algorithm for generic instances of Quantum Max Cut improves on the previously known best algorithms, which give approximation ratios of less than 0.54.
\end{abstract}

\section{Introduction}
Since Farhi, Goldstone, and Gutmann proposed the Quantum Approximate Optimization Algorithm (QAOA) \cite{FGG14}, parameterized quantum circuits have become a major tool to design quantum algorithms for optimization problems.
For example, the growing area of Quantum Variational Algorithms (VQA), where parameterized quantum circuits are optimized in a variational manner, is considered a leading candidate to deliver quantum advantage over classical computation using the current quantum machines \cite{VQA21}.

A parameterized quantum circuit is a quantum circuit whose components are determined by given circuit parameters $\vec\theta$ and output a state  $|\phi(\vec{\theta})\rangle$ in the end.
Given a parameterized circuit, the challenge in fully taking advantage of it lies in finding a good set of parameters in a non-convex landscape that can grow exponentially to the number of parameters. 
VQAs attempt to find locally optimal parameters, without guarantees on how good the output states are for the given optimization problem.
Other parameterized quantum circuit algorithms provide rigorous guarantees on the quality of their output states, often in the form of approximation ratios.
To make the search for good parameters more tractable, existing rigorous algorithms optimize the parameters over a small portion of the entire parameter space, by assuming the parameters to be equal on each circuit layer or on the entire circuit \cite{BBT07, GK11, GK12, BH16, BGKT19, GP19, HLP20, AGM20,KCM21, PT21a, PT21b,AGMS21,HNPTW21}.
As another limitation, the existing rigorous algorithms only work on specific instances such as those with regular interaction graphs or with bounded light-cone sizes.

In this paper, we introduce a novel way to optimize circuit parameters rigorously and efficiently without reducing the parameter space or limiting the algorithm to special case instances.
We do so by applying the ideas of the seminal Goemans-Williamson algorithm \cite{GW95}; for each edge of the interaction graph, we round an SDP solution vector attached to it to a circuit parameter on the edge.
In this way, polynomially many circuit parameters can be efficiently and rigorously optimized.
The parameterized quantum circuit we optimize in this study was initially proposed by Anshu, Gosset, and Morenz \cite{AGM20} to approximate degree-bounded regular instances of Quantum Max Cut.
We obtain a 0.562-approximation algorithm for generic instances of Quantum Max Cut, improving on the previously known best approximation algorithms for the problem, which give approximation ratios of less than 0.54 \cite{AGM20, PT21a}.

Quantum Max Cut is a QMA-hard quantum optimization problem that has recently been studied in the context of approximation \cite{GP19, AGM20, PT21a, PT21b,KP22, PT22}. 
The problem asks one to find the largest eigenvalue $\lambda_{max}(H)$ of a 2-local Hamiltonian 
defined by the local term $H_{ij}= (I- X_i X_j - Y_i Y_j- Z_i Z_j)/4 =(|01\rangle - |10\rangle)_{ij}(\langle01 |-\langle10 |)_{ij}/2$.
Note that the local term is a projector onto the singlet state, giving it the name Quantum Max Cut.
More formally, we give the following definition. 
\begin{definition}[Pauli matrices]\label{def:q_max_cut}
We use the standard definition of the Pauli matrices:
\begin{align*}
I =
\begin{pmatrix}
1 & 0\\
0 & 1
\end{pmatrix}, \ 
X =
\begin{pmatrix}
0 & 1\\
1 & 0
\end{pmatrix},\
Y =
\begin{pmatrix}
0 & -i\\
i & 0
\end{pmatrix},\
Z =
\begin{pmatrix}
1 & 0\\
0 & -1
\end{pmatrix}.
\end{align*}
An indexed Pauli matrix works on a specific qubit in a multi-qubit system.
Formally, 
$P_i :=I \otimes  \dots \otimes I \otimes P \otimes I \otimes \dots I$, where the matrix $P \in\{X,Y,Z\}$ is in the $i$-th place.  
\end{definition}
\begin{definition}[Quantum Max Cut]
Given a weighted graph $G = (V, E, W),$ find the maximum eigenvalue of the following Hamiltonian:
\begin{align*}
    H = \sum_{(i,j)\in E} w_{ij} H_{ij}, \label{eq:qmc} \tag{$H$}
\end{align*}
where $H_{ij} = \frac{1}{4}(I- X_i X_j - Y_i Y_j - Z_i Z_j)$, and $ w_{ij} \ge 0$ for all $(i,j)\in E$. We denote the maximum eigenvalue as $\text{OPT}:=\lambda_\text{max}(H)$.
\end{definition}

Because Quantum Max Cut is QMA-hard \cite{PM17}, the exact optimum is not expected to be found efficiently. 
As a natural reaction to this hardness, multiplicative approximation algorithms have been studied for the problem.
An $\alpha$-approximation algorithm returns a state $|\phi\rangle$ with its objective value $\langle\phi|H|\phi\rangle$ guaranteed to be high relative to the optimum, satisfying $\langle\phi|H|\phi\rangle \ge\alpha\lambda_{max}(H)$ for a constant $\alpha\in[0,1]$.

Given that the Goemans-Williamson algorithm \cite{GW95} approximates Max Cut optimally up to the unique games conjecture \cite{KKEO07} and that the evaluation of 2-local Hamiltonians and the cut of graphs share the similarity of being quadratic, it is natural to adopt the Goemans-Williamson algorithm in some form to approximately optimize Quantum Max Cut.
This path was indeed followed by all existing approximation algorithms for Quantum Max Cut \cite{GP19, AGM20, PT21a, PT21b,KP22, PT22} and for the Max 2-local Hamiltonian with positive-semidefinite (PSD) terms \cite{HLP20,PT21b, PT22}, which contains Quantum Max Cut as a sub-problem.

To apply the ideas from the Goemans-Williamson algorithm \cite{GW95} to Quantum Max Cut, two things are required: an SDP that represents the optimization problem, and an algorithm to round the SDP solution to a quantum state.

In the case of the Quantum Max Cut and Max 2-local Hamiltonian, a natural set of SDPs called the quantum Lasserre hierarchy \cite{ANP08} has been used by existing algorithms.
Roughly speaking, a level-$k$ SDP in the hierarchy forces its feasible solutions to be consistent with quantum states up to $k$-qubit correlation.
Because the Lasserre SDPs of a higher level are closer to the physical reality, they may yield better results for a given optimization problem.
Some algorithms \cite{GP19, HLP20, AGM20, PT21b} employ level-1 SDPs while others employ \cite{PT21a, PT22} level-2 SDPs. 
The recent Unique Games-hardness result \cite{HNPTW21} of  0.956-approximation for Quantum Max Cut was also proved using a level-1 Lasserre SDP.
The algorithms using the level-2 SDPs \cite{PT21a, PT22} demonstrate that approximation ratios are improved compared to the level-1 SDP algorithms.
Especially, \cite{PT21a} showed that a monogamy of entanglement relation can be derived from the SDP solutions, which we also use in this paper. 
The SDP we use in this paper is of level-2, inspired by and resembling the SDPs in \cite{PT21a,PT22}. 

We now give a formal definition of the quantum Lasserre hierarchy.
\begin{definition}[$\mathcal P_n(k)$]
Given a natural number $k$, we define the set $\mathcal{P}_n(k)$ as the set of Pauli matrices  with a Hamming weight less than or equal to $k$ on $n$ qubits:
\begin{align*}
    \mathcal{P}_n(k) =\{ \prod_{i\in A} P_i | A\subset [n], |A|\le k, P_i\in\{X_i, Y_i, Z_i\}\} .
\end{align*}
\end{definition}

\begin{definition}[Canonical Lasserre]
Given a weighted graph $G = (V,E, W),$ we define the level-$k$ Lasserre $L_k(G)$ for Quantum Max Cut on the graph $G$ as the following SDP:

\begin{align*} \label{eq:lasserre} \tag{$L_k$}
\text{Maximize}\ \sum_{(i,j) \in E} \frac{1}{4}w_{ij } (M(I,I) - M(X_i X_j, I) - M(Z_i Z_j, I) - M(Z_i Z_j, I)) \ \ \ \ \
\end{align*}
\begin{align*}
\text{subject to:}\ &\\
& M(\Phi, \Phi) = 1,  &\forall \ \Phi\in \mathcal P_n(k) \\
& M(\Phi, \Psi) =0,   &\forall \ \Phi,\Psi \in \mathcal P_n(k) \text{ s.t. } \Phi\Psi + \Psi\Phi =0  \\
& M(\Phi_1, \Psi_1) = M(\Phi_2, \Psi_2),   &\forall \  \Phi_1,\Psi_1,\Phi_2,\Psi_2 \in \mathcal P_n(k) \text{ s.t. } \Phi_1\Psi_1 = \Phi_2 \Psi_2\\
& M(\Phi_1, \Psi_1) = -M(\Phi_2, \Psi_2),   &\forall \ \Phi_1,\Psi_1,\Phi_2,\Psi_2 \in \mathcal P_n(k) \text{ s.t. } \Phi_1\Psi_1 = -\Phi_2 \Psi_2\\
& M \succeq 0, \\
& M \in \mathbb{R}^{|\mathcal P_n(k)| \times |\mathcal P_n(k)|}.
\end{align*}
\end{definition}

\begin{lemma}\label{lem:q_lasserre}
The level-$k$ Lasserre $L_k(G)$ is a relaxation of Quantum Max Cut for any $k\in [n]$. 
Meaning, for any $k\in[n]$ and an arbitrary $n$-qubit state $|\phi\rangle$,
there exists a matrix $M$ that satisfies the constraints of $L_k(G)$ and its objective value for $L_k(G)$ is at least $\langle\phi|H|\phi\rangle$.
\end{lemma}
\begin{proof}
See Appendix \ref{appendix:q_lasserre}.
\end{proof}

The other component of the Goemans-Williamson scheme is a rounding algorithm that maps SDP solutions to valid quantum states.
The existing rounding algorithms for Quantum Max Cut and Max 2-local Hamiltonian round to product states of 1- or 2-qubit states, with the exception of \cite{AGM20}.
Anshu, Gosset, and Morenz \cite{AGM20} devised a shallow parameterized quantum circuit to generate high energy states when given instances with low-degree regular interaction graphs. 
Our rounding algorithm also makes use of the circuit proposed in \cite{AGM20}, but we improve it by allowing parameters on each edge to be different from each other and to be applied to generic instances.
Because there are much more degrees of freedom in our circuits than in \cite{AGM20}, it is more challenging to optimize the circuits.

The circuit is comprised of mutually commuting gates $\exp[i \theta_{ij} P(i)P(j)]$ on each edge $(i,j)\in E$, where $P(i)$ is a Pauli matrix on the qubit $i$, and $\theta_{ij}$ is a real parameter.
The gates together implement a unitary $U(\vec\theta) = \Pi_{(i,j)\in E} \exp[i \theta_{ij} P(i)P(j)]$.
Our algorithm first randomly rounds the level-1 SDP solution vectors to a random bit string $|z\rangle$.
Then the bit string is evolved by the unitary $U(\vec\theta)$, with the circuit parameters $\vec\theta$ determined by the level-2 SDP solution vectors. 

To illustrate how the circuit works, suppose an initial state $|z\rangle_{i,j} = |01\rangle$ is assigned on the qubits $i,j$.
Our algorithm assigns $P(i) = Y_i, P(j) = X_j$ based on the values assigned to $z_i$ and $z_j$.
The gate on $(i,j)$ will evolve the bit string to the state \begin{align*}
    \exp[i \theta_{ij} Y_i X_j]|01\rangle = \cos \theta_{ij} |01\rangle -\sin\theta_{ij}|10\rangle.
\end{align*}
The state yields the maximum energy 1 on the edge if $\theta_{ij} = \pi/2$, and energy 1/2 if $\theta = \pi/2$.
Thus, we assign a higher value to $\theta_{ij}$ if the SDP vector on $(i,j)$ suggests that the state should be more entangled on the edge.
If any of the vertices $i,j$ is connected to edges than other $(i,j)$, the energy on $(i,j)$ will be affected by the parameters on the neighboring edges as well.
This quantum interference poses a great challenge to the design of the algorithm, and we circumvent it by utilizing the monogamy of entanglement on each vertex.

The algorithm we propose approximates Quantum Max Cut with the ratio 0.562, outperforming existing algorithms.
More important are the novel techniques we utilize to optimize quantum circuits by consulting SDP vectors. 
It is noteworthy that our algorithm outperforms the parameterized circuit algorithm in \cite{AGM20} even though ours works on more general instances. 
This perhaps illustrates the power of circuit optimization based on SDP rounding. 
We believe, however, that our algorithm is not optimal in terms of approximation ratio.
According to Brand\~ao and Harrow \cite{BH16} and Parekh and Thompson \cite{PT22}, an optimal algorithm should be able to generate arbitrary product states, but our algorithm cannot.
More specifically to this problem, the Hamiltonian is symmetric under a unitary $U\otimes U\otimes\dots\otimes U$ for any 1 qubit unitary $U,$ but the algorithm is not symmetric under the same transformation.

In conclusion, the algorithm lies in a narrow intersection of two important classes of algorithms: parameterized quantum circuit algorithms and SDP rounding algorithms. The connection between the two areas provides two different ways to view our algorithm: it is the first quantum SDP rounding algorithm that works on generic Quantum Max Cut instances, and it is the first algorithm that optimize circuit parameters via SDP rounding.

Many natural questions arise. Can we optimize QAOA circuits using SDP? Can SDPs help optimize circuits for practical problems such as the quantum chemistry problems? Can we get a better rounding circuit by optimizing the generalized version \cite{AGMS21} of the circuit we used?
Will level-3 SDPs be helpful for 2-local Hamiltonian problems?
Is the approximation ratio upper bound given in \cite{HNPTW21} optimal?
Can we optimize non-commuting circuits using SDPs?

\section{Our SDP}
In this section, we define a customized level-2 Lasserre SDP for the sake of better notation.
An equivalent way to describe SDP (\ref{eq:lasserre}) is to use a set of vectors as variables of an optimization program instead of a PSD matrix itself. 
This is because given a PSD matrix, one can find a set of vectors of which the Gram matrix is the given PSD matrix (e.g. via the Cholesky decomposition),
and conversely given a set of vectors, their Gram matrix is always PSD. 

We choose to represent the SDP in vectors because the geometric meaning is more accessible in this way. Furthermore, for indexing reasons, we only leave the necessary constraints of 
the previous SDP (\ref{eq:lasserre}),
regarding the $I, X_i, Y_i, Z_i, X_i X_j, Y_i Y_j, Z_i Z_j$ terms.
 
 \begin{definition}
Given a weighted graph $G = (V,E, W),$ we define SDP ($S$) for Quantum Max Cut on the graph $G$ as followes:

\begin{align}\label{eq:sdp} \tag{$S$}
\text{Maximize}\ & \sum_{(i,j) \in E} \frac{1}{4}w_{ij} (v_0 -v_{ij,1} - v_{ij,2} -v_{ij,3})\cdot v_0,\\
\text{subject to:}\ & \nonumber\\
& \|v_0\| = 1,   \\
& \|v_{i,a}\| =1, \\
& v_{i,a}\cdot v_{i,b} = 0, \\
& \|v_{ij,a}\| =1, \label{eq:sdp_const_4}\\ 
& v_{i,a} \cdot v_{j,a} = v_{ij,a}\cdot v_0 ,\\
& v_{ij,a} \cdot v_{jk,a}= v_{ik,a}\cdot v_0, \label{eq:sdp_const_6}\\
& v_{ij,a}\cdot v_{jk,b} =0, \label{eq:sdp_const_6_} \\
& v_{ij,a}\cdot v_{ij,b} = -v_{ij,c}\cdot v_0, \label{eq:sdp_const_7}\\
& v_0,\ v_{i,a},\ v_{ij,a} \in \mathbb{R}^{N}, \nonumber
\end{align}
\begin{align*}
\forall \text{ distinct } i,j,k \in V,  \ \ \ \forall\text{ distinct } a,b,c\in \{1,2,3\},  \\
\text{ for a sufficiently big $N \in O(|V|^2)$}.\\
\end{align*} 
We denote the optimal value of the above SDP as OPT$_\text{SDP}$.
 \end{definition}

\begin{lemma}
The SDP (\ref{eq:sdp}) is a relaxation of the level-2 Lasserre ($L_2$). 
Thus, given an arbitrary matrix $M$ that satisfies the constraints of the SDP ($L_2$),
there exists a tuple of vectors $(v_0, (v_{i,a})_{i,a}, (v_{ij,a})_{i,j,a})$ that satisfies the constraints of (\ref{eq:sdp}), and its objective value for (\ref{eq:sdp}) is at least the objective value of $M$ for ($L_2$).
\end{lemma}
\begin{proof}
Let us consider an SDP with the same objective function as SDP (\ref{eq:sdp}), and a subset of the constraints restricted to $T:=\{I, X_i, Y_i, Z_i, X_i X_j, Y_i Y_j, Z_i Z_j| i,j \in V, i\ne j\}\subset \mathcal P_n(2)$.
Because we now have fewer constraints, we get a relaxed SDP with respect to $L_2$. 
The constraints we get by restricting the indices are
\begin{align*}
&M(I,I) =1, \\
&M(P_i, P_i)  =1, \\
& M(P_i, Q_i) =0, \\
& M(P_i P_j, P_i P_j) = 1, \\
& M(P_i, P_j) = M(P_i P_j, I), \\
& M(P_i P_j, P_j P_k) = M(P_i P_k,I), \\
& M(P_i P_j, Q_j Q_k) = 0,\\
& M(P_i P_j, Q_i Q_j) = -M(R_i R_j,I),\\
& M \in \mathbb{R}^{|\mathcal P _n(2)| \times |\mathcal P_n(2)|},\\
& M \succeq 0, \\
& \forall \text{ distinct } i,j,k\in V, \ \forall \text{ distinct } P, Q, R\in \{X, Y, Z\}.
\end{align*}

Since $M$ is a PSD matrix, we can obtain a Cholesky decomposition $(w_s)_{s\in T}$ such that $w_s\cdot w_t =M(s,t)$, where $w_s, w_t \in \mathbb{ R}^N$ for all $s,t\in T$, with some natural number $N$. Now re-index the vectors so that 
\begin{align*}
    &v_0 = w_I, \\
    &v_{i,1} = w_{X_i}, \ v_{i,2} = w_{Y_i}, \ w_{i,3} = w_{Z_i}, \\
    &v_{ij,1} = w_{X_i X_j}, \ v_{ij,2} = w_{Y_i Y_j}, \ v_{ij,3} = w_{Z_i Z_j} 
\end{align*}
for all distinct $i,j \in V.$
Writing the above SDP constraints and the objective function in terms of the vectors $v_0, v_{i,a}, v_{ij,a}$ gives the desired SDP (\ref{eq:sdp}).
\end{proof}
Feasible solutions of the SDP (\ref{eq:sdp}) have useful properties that we will use later.

\begin{lemma}[Useful properties of SDP solutions]\label{lem:properties_of_sdp_sol}
Let $(v_0, (v_{i,a})_{i,a}, (v_{ij,a})_{i,j,a})$ be a feasible solution to the SDP (\ref{eq:sdp}).
Define a vector 
\begin{align}
    v_{ij}:=v_{ij,1} +v_{ij,2} +v_{ij,3} \label{eq:def_vij}
\end{align}
on each edge $(i,j)$. Then the following relations hold for all $i,j \in V$:
\begin{align}
&\|v_{ij}\|^2 = 3-2v_{ij}\cdot v_0. \label{eq:sdp_property_1} \\
&\|v_0 +v_{ij}\|^2 =4, \label{eq:sdp_property_2} \\
&v_{ij} \cdot v_{jk} =v_{ik}\cdot v_0. \label{eq:sdp_property_3}
\end{align}
In particular, $v_0+v_{ij}$ is on a sphere.
\end{lemma}

\begin{proof}
We derive the equations from the constraints of the SDP (\ref{eq:sdp}).
For the equation (\ref{eq:sdp_property_1}), expand the left-hand side to get
\begin{align*}
    \|v_{ij}\|^2 &= \|v_{ij,1} +  v_{ij,2} +v_{ij,3}\|^2 = \sum_{a\in[3]} \|v_{ij,a}\|^2 + 2(v_{ij,1} \cdot v_{ij,2} +v_{ij,2} \cdot v_{ij,3} +v_{ij,3} \cdot v_{ij,1}) \\
    &=3 - 2(v_{ij,1} +v_{ij,2} +v_{ij,3} )\cdot v_0 \\ 
    &= 3-2v_{ij}\cdot v_0,
\end{align*}
where the second line is from the SDP constraints (\ref{eq:sdp_const_4}) and (\ref{eq:sdp_const_7}). 
We get the equation (\ref{eq:sdp_property_2}) using the equation (\ref{eq:sdp_property_1}):
\begin{align*}
    \|v_0 + v_{ij}\|^2 = v_0\cdot v_0 + 2v_0\cdot v_{ij} + v_{ij}\cdot v_{ij} = 
    1 + 2 v_{ij} \cdot v_0 + 3-2v_{ij}\cdot v_0 = 4.
\end{align*}
To get the equation (\ref{eq:sdp_property_3}),
\begin{align*}
    v_{ij}\cdot v_{jk} = \sum_{a\in[3]} v_{ij,a} \cdot \sum_{b\in[3]} v_{jk,b}
    = \sum_{a\in[3]} v_{ij,a}\cdot v_{jk,a} = \sum_{a\in [3]} v_{ik,a} \cdot v_0 = v_{ik}\cdot v_0. 
\end{align*}
Here we used the SDP constraints (\ref{eq:sdp_const_6}) and (\ref{eq:sdp_const_6_}). 
\end{proof}

The following lemma is crucial for the design of the algorithm. 
It says that the total SDP value on $d$ edges with weight 1 that share a common vertex is at most $(d+1)/2$. 
For context, a uniformly random state has energy $d/4$, and a best bit string would yield energy $d/2.$
So the lemma says that the quantum advantage in energy per edge decreases as the degree increases.
This lemma is from \cite{PT21a} (and a quantum state version of it from \cite{AGM20}) and we prove it in a more elementary way without the Schur complement argument used in the original proof. 

\begin{lemma}[Monogamy of entanglement, \cite{PT21a} ] \label{lem:monogamy}
Given a feasible solution $(v_0, (v_{i,a})_{i,a}, (v_{ij,a})_{i,j,a})$ to SDP (\ref{eq:sdp}), a vertex $i\in V$, and a set $T$ of adjacent vertices to the node $i$,
the total contributions to the value of SDP (\ref{eq:sdp}) from the star graph formed by $i$ and $T$ is upper bounded by $(d+1)/2$, 
where $d$ is the number of vertices in $T$. 
In other words,
\begin{align}
    \frac{1}{4}\sum_{j\in T} (v_0 -v_{ij})\cdot v_0\le \frac{d + 1}{2}, \label{ineq:monogamy} 
\end{align}
where $(i,j)\in E \forall  j\in T$ and as before, $ v_{ij}=v_{ij,1} +v_{ij,2} +v_{ij,3}.$ 
\end{lemma}

\begin{proof}
Without loss of generality, we can assume that $T=[d]=\{1,2,\dots,d\}$.
Then, the lemma is equivalent to the inequality
\begin{align*}
    \sum_{j\in [d]} v_{ij}\cdot v_0 \ge -d-2.
\end{align*}
The Gram matrix from the vectors $\{v_0, v_{i1},v_{i2},\dots, v_{id}\}$ is PSD, so we get
\begin{align*}
0\preceq
&    \begin{pmatrix}
    1 & v_0\cdot v_{i1} & \cdots & v_0\cdot v_{id} \\
    v_0\cdot v_{i1} & \|v_{i1}\|^2 & \cdots & v_{i1}\cdot v_{id} \\
    \vdots & \vdots & \ddots & \vdots\\
    v_0\cdot v_{id} & v_{i1}\cdot v_{id} & \dots & \|v_{id}\|^2 
    \end{pmatrix}\\
    \\
=
  &  \begin{pmatrix}
    1 & v_0\cdot v_{i1} & \cdots & v_0\cdot v_{id} \\
    v_0\cdot v_{i1} & 3- 2v_0\cdot v_{i1} & \cdots & v_{0}\cdot v_{1d} \\
    \vdots & \vdots & \ddots & \vdots\\
    v_0\cdot v_{id} & v_{0}\cdot v_{1d} & \dots & 3- 2v_0 \cdot v_{id} 
    \end{pmatrix},\\
\end{align*}
where the second line is from the equations (\ref{eq:sdp_property_2}) and (\ref{eq:sdp_property_3}).

Since an average of PSD matrices is a PSD matrix,
we average the above matrix over the permutation of the indices $\{1,2,\dots,d\}$ to get
\begin{align*}
    0\preceq
  &  \begin{pmatrix}
    1 & s & \cdots &s \\
    s & 3- 2s & \cdots & t \\
    \vdots & \vdots & \ddots & \vdots\\
    s & t & \dots & 3- 2s 
    \end{pmatrix},
\end{align*}

where
\begin{align*}
    s:= \mathop{\mathbb{E}}_{k\in [d]} \ [v_0 \cdot v_{ik}],\ \ \ t:=\mathop{\mathbb{E}}_{\substack{k,l\in [d] \\ k\ne l}}[v_0 \cdot v_{kl}].
\end{align*}

We want to lower bound the variable $s$ to prove the lemma. To this end, we multiply by a vector $(x, 1,\dots, 1)$ from the both sides to the above matrix, where $x\in \mathbb{R}$.
This gives
\begin{align*}
    0\le x^2 + 2dsx + (3-2s)d + td (d-1).
\end{align*}
This inequality should hold regardless of the value of $x$. Therefore,
\begin{align*}
    0\ge (ds)^2 - d(3-2s +t(d-1))= d^2 s^2 +2ds  -d(3+td-t).
\end{align*}

The right-hand side expression is a concave parabola in $s$, and the constant term is a non-increasing function of $t$. 
This means that increasing $t$ can only widen the range that $s$ could be in.
From the property of SDP solutions (\ref{eq:sdp_property_1}), we know that $-3\le s,t \le 1$, so we set $t=1$ to lower bound $s$.
By finding the roots of the right-hand side expression, we get $-(d+2)/d \le s \le 1$.
Therefore, we get $\sum_{j\in T} v_0\cdot v_{ij} =ds \ge -d-2$, proving the lemma.
\end{proof}

\section{Approximation algorithm}
In this section, we present a 0.562-approximation algorithm for the quantum analogue of Max Cut defined in Definition \ref{def:q_max_cut}, and we prove its approximation ratio. 
In this paper, the value of the arccosine function is always in $[0,\pi]$.

\begin{algorithm} 
\caption{0.562-approximation algorithm}\label{alg:alg}
\begin{algorithmic}[1]
    \State Given a weighted graph $G = (V,E,W)$,
    \State Solve the SDP (\ref{eq:sdp}), and get vectors $(v_0,(v_{i,a})_{i,a},(v_{ij,a})_{ij,a})$ that optimize the SDP.
    \State Generate a random bit string $|z\rangle\in \{0,1\}^n$ using the Goemans-Williamson rounding as follows:
    \State \indent Pick a uniformly random $a\in \{1,2,3\}$ and a uniformly random vector $r\in S^{N-1}$.
    \State \indent Assign $z_i = (\sgn(v_{i,a} \cdot r)+1)/2$, for all $i\in V$.
    \State Generate a parameterized state $|\phi(\vec\theta)\rangle$ from the bit string $|z\rangle$ as follows:
    \State \indent Compute a vector $v_{ij}:= v_{ij,1} + v_{ij,2} + v_{ij,3}$ on each $(i,j)\in E$.
    \State \indent Compute a circuit parameter $\theta_{ij} = f(\gamma_{ij})$ on each edge $(i,j)\in E$,
    \State \indent where $\gamma_{ij}:= -\frac{(v_0 + v_{ij})\cdot v_0}{\|v_0 +v_{ij}\|\|v_0\|}$, $f(\gamma):=\cos^{-1} \exp [-\alpha_0 \max( \gamma,0)]/2$, $\alpha_0:= 0.041.$
    \State \indent Compose a circuit $U(\vec{\theta}) = \prod\limits_{(i,j)\in E} \exp(i \theta_{ij} P(i) P(j))$, where $P(i) = X_i$ if $z_i = 1$, and \indent $P(i) = Y_i $ if $z_i =0$.
    \State \indent Evolve the bit string $|z\rangle$ on the circuit and get $|\phi(\vec\theta)\rangle = 
    U(\vec{\theta})|z\rangle$.
    \State Output the state $|\phi(\vec\theta)\rangle$.    
\end{algorithmic}
\end{algorithm}

The algorithm design is within the Goemans-Williamson framework \cite{GW95}.
We first solve the SDP (\ref{eq:sdp}).
Our rounding algorithm then maps the SDP solution to a quantum state in two stages.
In the first stage, a bit string $|z\rangle\ \in\{0,1\}^n$ is computed such that
on each edge $(i,j)\in E$, the bits $z_i$ and $z_j$ are likely to be assigned different values
if the level-1 vectors on the vertices are far from each other. 
When two different bit values are assigned to the two vertices of an edge, we say the edge is ``cut''.
The probability that an edge $(i,j)$ is cut by $|z\rangle$ is lower bounded in Lemma \ref{lem:pr_cut}.

In the second stage, the bit string is entangled by a parameterized quantum circuit. 
The circuit consists of mutually commuting 2-qubit gates $\exp[i\theta_{ij} P(i) P(j)]$ on each edge $(i,j)\in E$, 
where $\theta_{ij}$ is a real parameter, and $P(i)=Y_i$ if $z_i =0,$ and $P(i)=X_i$ if $z_i=1$.
For example, say $|z_i z_j\rangle =|01\rangle$ was assigned to $(i,j)\in E$ in the first stage.
The gate would evolve the state to 
\begin{align*}
    \exp[i \theta_{ij} Y_i X_j]|01\rangle = \cos \theta_{ij} |01\rangle -\sin\theta_{ij}|10\rangle.
\end{align*}
The energy on the edge is an increasing function of $\theta_{ij} $ in $[0,\pi/4]$.
Thus, the rounding algorithm will assign a higher value to $\theta_{ij}$ if the SDP value on the edge $(v_0- v_{ij})\cdot v_0$ is higher.

The analysis becomes more complicated when there are other gates connected to $(i,j)$ to interfere.
The monogamy of entanglement relation stated in Lemma \ref{lem:monogamy}
implies that the assignment state cannot be too strongly entangled on all edges, meaning the neighboring parameters should be considered as well for optimization on $(i,j)$. 
We optimize the circuit parameter against the monogamy of entanglement via the function $f$.

\begin{lemma}\label{lem:pr_cut}
Suppose an SDP solution $(v_0,(v_{i,a})_{i,a}, (v_{ij,a})_{ij,a})$ and a bit string $|z\rangle$ are as in the statement of Algorithm \ref{alg:alg},
and a vector $v_{ij}$ is defined by the SDP solution as $v_{ij}=v_{ij,1}+v_{ij,2}+v_{ij,3}$. 
Let a real number $\gamma_{ij}:= -\frac{(v_0 + v_{ij})\cdot v_0}{\|v_0 + v_{ij}\|\|v_0\|} $ be the normalized inner product of the vectors $v_0 + v_{ij}$ and $-v_0$.
Then the probability that an edge $(i,j)\in E$ is cut by the string $|z\rangle$ is lower bounded as
\begin{align}
\Pr[z_i\ne z_j] \ge \frac{\alpha_{GW}}{3}( 2 +  \gamma_{ij}),
\end{align}
where the coefficient $\alpha_{GW}$ is the Goemans-Williamson constant 
\cite{GW95}
\begin{align*}
    \alpha_{GW} := \min_{t \in [-1,1]} \frac{1}{\pi}\frac{\arccos  t}{1/2 -t/2} =0.8785\cdots.
\end{align*}
\end{lemma}

\begin{proof}
See Appendix \ref{appendix:pr_cut}.
\end{proof}

The following lemma deals with the performance of the parameterized circuit $U(\vec\theta)$ constructed in the algorithm. The parameterized family of circuits was proposed in \cite{AGM20}, with the assumption that the parameters are the same across the edges.
We modify the analysis of \cite{AGM20} to include the case that we need, the case where different parameters can be assigned to different edges.

\begin{lemma}[Energy contribution from $(i,j)$]\label{lem:energy_on_edge}
Let a quantum state $|\phi(\vec\theta)\rangle$ and a bit string $|z\rangle$ be as in the statement of Algorithm \ref{alg:alg}, and circuit parameters 
$\theta_{ij}\ge 0$ for all $(i,j)\in E$. 
Then we can lower bound the expected energy contribution from the edge $(i,j)$ as
\begin{align}
    \langle\phi(\vec\theta)|4H_{ij}|\phi(\vec\theta)\rangle \ \ge \ 
    \begin{cases}
    1 + 
    \sin(2\theta_{ij}) ( A_{ij} + B_{ij})  + A_{ij} B_{ij} & \text{if } z_i \ne z_j, \\
    0 & \text{if } z_i = z_j,
    \end{cases} \label{eq:energy_ij}
\end{align}
where 
\begin{align*}
A_{ij}:=    \prod_{k\in N(i)\setminus \{j\}} \cos (2\theta_{ik}), \ \  \ \
B_{ij}:=    \prod_{k\in N(j) \setminus \{i\}} \cos (2\theta_{kj}).
\end{align*}
and $N(i):= \{j\in V| (i,j)\in E\}$
for all $i\in V$.
\end{lemma}

\begin{proof}
See Appendix \ref{appendix:energy_on_edge}.
\end{proof}

\begin{theorem}[Main] \label{thm:main}
Algorithm \ref{alg:alg} is a 0.562-approximation algorithm for Quantum Max Cut defined in the Definition \ref{def:q_max_cut}, with the output state $|\phi(\vec\theta)\rangle$ having energy $\langle\phi(\vec\theta)|H|\phi(\vec\theta)\rangle \ge 0.562 \text{OPT}$.
\end{theorem}
\begin{proof}
We will show that the expected energy of the random state $|\phi(\vec\theta)\rangle$ on an arbitrary edge $(i,j)\in E$ is at least the SDP contribution from that edge, $w_{ij}(v_0-v_{ij})\cdot v_0/4$, times 0.562. 
Summing over all the edges, this will give us $\mathbb{E}[\langle\phi(\vec\theta)|H|\phi(\vec\theta)\rangle] =\sum_{(i,j)\in E} \mathbb{E}[\langle\phi(\vec\theta)|w_{ij} H_{ij}|\phi(\vec\theta)\rangle]
\ge 0.562 \sum_{(i,j)\in E} w_{ij}(v_0-v_{ij})\cdot v_0/4 
= 0.562 \text{OPT}_\text{SDP}\ge 0.562 \text{OPT}$, proving the theorem.

Now, let us focus on a single edge $(i,j)\in E.$ 
The expected energy of the state $|\phi(\vec\theta)\rangle$ on the edge $(i,j)$ is $\mathbb{E}[\langle\phi(\vec\theta)|w_{ij} H_{ij}|\phi(\vec\theta)\rangle]$, whereas the SDP value contribution is $w_{ij} (v_0 -v_{ij}) \cdot v_0/4$. Therefore we want to lower bound their ratio

\begin{align}
\alpha_{ij} := \frac{\mathbb{E}\left[\langle\phi(\vec\theta)| 4H_{ij}|\phi(\vec\theta)\rangle\right]}{(v_0 -v_{ij}) \cdot v_0}. \label{eq:main_def_alphaij}    
\end{align}
The edge $(i,j)$ is randomly cut by the bit string $|z\rangle$ in the algorithm statement, 
and we can lower bound the energy by an expression that depends on whether the edge $(i,j)$ is cut or not. 
By applying Lemmas \ref{lem:pr_cut} and \ref{lem:energy_on_edge}, we get
\begin{align}
    \mathbb{E}[\langle\phi(\vec\theta)|4 H_{ij}|\phi(\vec\theta)\rangle] &\ge
    \Pr[z_i\ne z_j]\big[ 1 + 
    \sin(2\theta_{ij}) ( A_{ij} + B_{ij})  + A_{ij} B_{ij}\big]     +\Pr[z_i = z_j] \cdot 0 \label{eq:main_bound_edge_en}\\
    &\ge 
    \frac{\alpha_{GW}}{3}( 2 +  \gamma_{ij})\big[ 1 + 
    \sin(2\theta_{ij}) ( A_{ij} + B_{ij})  + A_{ij} B_{ij}\big], \label{eq:main_bound_edge_cut_pr}
\end{align}
where the variables $A_{ij}, B_{ij}, \gamma_{ij}$ are as defined in the Lemmas; 
\begin{align*}
A_{ij}&:=    \prod_{k\in N(i)\setminus \{j\}} \cos (2\theta_{ik}), \ \  \ \
B_{ij}:=    \prod_{k\in N(j) \setminus \{i\}} \cos (2\theta_{kj}), \\
     \gamma_{ij} &:= -\frac{(v_0 + v_{ij})\cdot v_0}{\|v_0 + v_{ij}\|\|v_0\|} . 
\end{align*}

We have
\begin{align}
  \gamma_{ij} =  -\frac{(v_0 + v_{ij})\cdot v_0}{\|v_0 + v_{ij}\|\|v_0\|} = - \frac{1+ v_{ij}\cdot v_0}{2} \label{eq:main_gamma}
\end{align}
by the definition of $\gamma_{ij}$, so
we get $(v_0 - v_{ij})\cdot v_0 = 2(1+ \gamma_{ij})$.
Thus, we can lower bound the right-hand side of (\ref{eq:main_def_alphaij}) as
\begin{align}
    \alpha_{ij} \ge \frac{\alpha_{GW}}{6}\big[1 +\sin(2\theta_{ij})(A_{ij} + B_{ij} ) +A_{ij} B_{ij}\big]\frac{2 +  \gamma_{ij}}{1 +  \gamma_{ij}}. \label{ineq:main_approx_ratio}
\end{align}

Case 1: $-1\le\gamma_{ij} \le 0$. The circuit parameter on the edge becomes $\theta_{ij} = \cos^{-1} \exp  0 /2= 0$. Therefore, we get 
\begin{align*}
    \alpha_{ij} \ge \frac{\alpha_{GW}}{6}\big[1  +A_{ij} B_{ij}\big]\frac{2 +  \gamma_{ij}}{1 +  \gamma_{ij}}.
\end{align*}
Note that the value of $\gamma_{ij}$ does not affect the $A_{ij} B_{ij}$ term, so the right-hand side is a decreasing function in $\gamma_{ij}$ in the region $\gamma_{ij} \le 0.$ This means that it suffices to look at the case where $\gamma_{ij} = 0$ which is a sub-case of Case 2.

Case 2: $0\le \gamma_{ij} \le 1$. We have $\theta_{ij} = f(\gamma_{ij}) = \cos^{-1}\exp[-\alpha_0 \max(\gamma_{ij},0)]/2 \ge 0$. We want to lower bound the terms $A_{ij}, B_{ij}$ to lower bound the expression on the right-hand side of (\ref{ineq:main_approx_ratio}). By the definition of the function $f$, we have
\begin{align*}
    A_{ij} &=  \prod_{k\in N(i)\setminus \{j\}} \cos (2\theta_{ik}) \\
    &=\exp\sum_{k\in N(i)\setminus\{j\}} \log\cos ( 2\theta_{ik}) \\
    &=\exp\sum_{k\in N(i)\setminus\{j\}} \log\cos ( 2f(\gamma_{ij})) \\
    &=\exp\sum_{k\in N(i)\setminus\{j\}} \log \cos  \cos^{-1} \exp [-\alpha_0\max( \gamma_{ik},0)] \\
    &=\exp\sum_{k\in N(i)\setminus\{j\}}  -\alpha_0\max( \gamma_{ik},0) \\
    &=\exp\left[-\alpha_0 \sum_{k\in N(i)^+ \setminus \{j\}}   \gamma_{ik}\right], 
\end{align*}
where the set $N(i)^+$ is defined to be $N(i)^+ := \{j | (i,j)\in E,\ \gamma_{ij} > 0\}.$
We will use the monogamy inequality (\ref{ineq:monogamy}) to lower bound the above expression.
From the equation (\ref{eq:main_gamma}), we have
\begin{align}
\sum_{k\in N(i)^+}  \gamma_{ik} &= -\sum_{k\in N(i)^+}\frac{1+ v_{ij}\cdot v_0}{2} = \sum_{k\in N(i)^+}\frac{-2 + (v_0 -  v_{ij})\cdot v_0}{2} \label{eq:main_aij1}\\
&\le -|N(i)^+| + |N(i)^+| +1 \label{eq:main_aij2}\\
&=1.\nonumber
\end{align}
Therefore we have
\begin{align}
    A_{ij}  &=\exp\left[-\alpha_0 \sum_{k\in N(i)^+ \setminus \{j\}}   \gamma_{ik}\right] \nonumber\\
    &=\exp\left[-\alpha_0 \left[\sum_{k\in N(i)^+}    \gamma_{ik} - \gamma_{ij}\right]\right]\nonumber\\
    &\ge \exp[-\alpha_0 (1 -\gamma_{ij})].\label{ineq:main_Aij_final}
\end{align}
Note that the product term $A_{ij}$ is now lower bounded by a function of a single variable $\gamma_{ij}.$
We can apply the same analysis to get an inequality
\begin{align}
B_{ij}\ge \exp[-\alpha_0 (1-\gamma_{ij})]. \label{ineq:main_Bij_final}
\end{align}
By plugging in the relations (\ref{ineq:main_Aij_final}) and (\ref{ineq:main_Bij_final}) to the RHS of (\ref{ineq:main_approx_ratio}), we get
\begin{align*}
    \alpha_{ij} &\ge \frac{\alpha_{GW}}{6}\big[1 +\sin(2\theta_{ij})(A_{ij} + B_{ij} ) +A_{ij} B_{ij}\big]\frac{2 +  \gamma_{ij}}{1 +  \gamma_{ij}}. \\
    &\ge \frac{\alpha_{GW}}{6}\left[1 + 2 \sin(\cos^{-1} e^{-\alpha_0  \gamma_{ij}}) e^{-\alpha_0 (1-\gamma_{ij})} + e^{-2\alpha_0 (1- \gamma_{ij})}\right] \frac{2+  \gamma_{ij}}{ 1+  \gamma_{ij}} \\
    &\ge \frac{\alpha_{GW}}{6}\left[1 + 2\sqrt{1- e^{-2\alpha_0  \gamma_{ij}}} e^{-\alpha_0 (1-\gamma_{ij})} + e^{-2\alpha_0 (1- \gamma_{ij})}\right] \frac{2+  \gamma_{ij}}{ 1+  \gamma_{ij}} \\
    &\ge \min_{\gamma\in[0,1]}\frac{\alpha_{GW}}{6}\left[1 + 2\sqrt{1- e^{-2\alpha_0  \gamma}} e^{-\alpha_0 (1-\gamma)} + e^{-2\alpha_0 (1- \gamma)}\right] \frac{2+  \gamma}{ 1+  \gamma} \\
    & = 0.562\cdots.
\end{align*}
This proves the theorem. The coefficient $\alpha_0$ is optimized so that the final ratio is maximized. 
\end{proof}

We make a few comments on the choice of the function $f$ because it is the most non-trivial part in the algorithm and the idea behind the function might find place in future designs of parameterized quantum circuits.

On the one hand, we have the monogamy of entanglement inequality which upper bounds the total SDP value on an arbitrary star graph. 
More specifically, it upper bounds the sum of the input parameters $\sum \gamma_{ik}$.
On the other hand, we want to lower bound a product of cosines of the output parameters $\prod \cos \theta_{ij}$ on the star graph.

Therefore, we want to somehow translate the linear condition into a multiplicative condition of cosines.
The most natural approach is to select a function $f$ such that the addition operator on the constraint side works \textit{homomorphically} on the energy lower bound side. 
In other words, if we have $ \cos(f(x))\cos(f(y)) =\cos(f(x +y))  $, then we can lower bound the right hand side by the monogamy inequality.
To this end, we set $f(x) = \cos^{-1} e^{ -ax}$ for a non-negative constant $a$.

To define the function in the negative domain, an important observation is that having a negative $\theta_{ij}$ is never better than having $\theta_{ij}=0$ in terms of the energy on $(i,j)$ or on its neighboring edges.
This observation leaves setting $\theta_{ij}=0$ in the negative domain as the only choice to keep $f$ non-decreasing.

\section{Acknowledgements}
The author appreciates the support from Post-doc Support Program from Sunkyunkwan University.

\appendix

\section{Omitted proofs} 
\subsection{Proof of Lemma \ref{lem:q_lasserre}} \label{appendix:q_lasserre}
\begin{proof}
Given a quantum state $|\phi\rangle$, define a matrix $M$ by matrix elements 
\begin{align*}
    M(\Phi,\Psi) :=\Re[ \langle\phi| \Psi\Phi|\phi\rangle] = \langle\phi|(\Phi\Psi + \Psi\Phi)|\phi\rangle/2
\end{align*} for all $\Phi, \Psi \in \mathcal P_{n}(k)$.

We will show that $M$ satisfies the constraints of the SDP (\ref{eq:lasserre}), and that the objective value of $M$ is equal to the energy of $|\phi\rangle$ to prove the lemma.
The first constraint is satisfied because of a Pauli identity $\Phi^2 =I$ for all Pauli $\Phi$, because we get $M(\Phi, \Phi) = \langle\phi|(I+I)|\phi\rangle/2 =1$. 
The second constraint comes directly from the definition of $M.$

For the other conditions, note that all the Pauli matrices are Hermitian. 
Therefore if $\Phi_1\Psi_1 = \Phi_2 \Psi_2$, we have $\Psi_1 \Phi_1 = (\Phi_1 \Psi_1)^\dagger = (\Phi_2 \Psi_2)^\dagger = \Psi_2 \Phi_2$. So the third constraint is satisfied, and the fourth for similar reasons.

To see that $M$ is PSD, observe that the matrix $M$ is real and symmetric because the Pauli matrices are Hermitian. 
Now, consider that for an arbitrary real vector $v \in \mathbb R ^{|\mathcal P_n(k)|},$ 
we have 
\begin{align*}
v^T M v & = \sum_{\Phi,\Psi}v_\Phi M(\Phi,\Psi)v_\Psi \\
& = \sum_{\Phi,\Psi}v_\Phi \langle\phi| (\Phi\Psi + \Psi\Phi)|\phi\rangle v_\Psi/2 \\
& =  \langle\phi| \sum_{\Phi,\Psi}(v_\Phi \Phi\Psi v_\Psi +  v_\Psi\Psi\Phi v_\Phi)|\phi\rangle/2 \\
& :=\langle\phi| \Phi' \Phi'|\phi\rangle = \|\Phi'|\phi\rangle\|^2\ge 0,
\end{align*}
where $\Phi':= \sum_\Phi v_\Phi \Phi$.

Finally to see that the objective value of $M$ is equal to the energy of $|\phi\rangle,$
\begin{align*}
\langle\phi|H|\phi\rangle &= \sum_{(i,j)\in E}\langle\phi|\frac{1}{4}w_{ij}(I- X_i X_j - Y_i Y_j - Z_i Z_j)|\phi\rangle \\
&= \sum_{(i,j)\in E}\frac{1}{4}w_{ij}( 1- \langle\phi|X_i X_j|\phi\rangle - \langle\phi|Y_i Y_j |\phi\rangle- \langle\phi|Z_i Z_j|\phi\rangle) \\ 
&=\sum_{(i,j) \in E} \frac{1}{4}w_{ij } (M(I,I) - M(X_i X_j, I) - M(Z_i Z_j, I) - M(Z_i Z_j, I)).
\end{align*}
Therefore the SDP (\ref{eq:lasserre}) is a relaxation of Quantum Max Cut.
\end{proof}

\subsection{Proof of Lemma \ref{lem:pr_cut}}\label{appendix:pr_cut}
\begin{proof}
We prove the lemma from the SDP constraints and the properties of the SDP solutions that are stated in Lemma \ref{lem:properties_of_sdp_sol}. We have

\begin{align}
\Pr[z_i \ne z_j] 
&= \frac{1}{3}\sum_{a\in \{1,2,3\}} \Pr_{r\in S^{N-1}}[\sgn(v_{i,a}\cdot r)=-\sgn(v_{j,a}\cdot r)]\label{eq:cut_lem_line_2}\\
&= \frac{1}{3} \sum_{a\in \{1,2,3\}} \frac{\arccos (v_{i,a}\cdot v_{j,a})}{\pi} \label{eq:cut_lem_line_3}\\
&=   \frac{1}{3} \sum_{a\in \{1,2,3\}}  \frac{\arccos (v_{i,a}\cdot v_{j,a})/\pi}{1/2 - v_{i,a}\cdot v_{j,a}/2}\left[\frac{1}{2} - \frac{v_{i,a}\cdot v_{j,a}}{2}\right]\nonumber \\
&\ge \
\frac{1}{3}\min_{t \in [-1,1]} \frac{\arccos (t)/\pi}{1/2 -t/2}
\sum_{a\in\{1,2,3\}} \left[\frac{1}{2} - \frac{v_{i,a}\cdot v_{j,a}}{2} \right] \nonumber\\
&=\frac{\alpha_{GW}}{3}\left[\frac{3}{2} -\frac{\sum_{a\in\{1,2,3\}} v_{ij,a}\cdot v_0 }{2}\right] \nonumber\\
&=\frac{\alpha_{GW}}{3}\left[ 2  - \frac{(v_0 + v_{ij})\cdot v_0}{2}\right]  \nonumber\\
&=\frac{\alpha_{GW}}{3}\left[ 2  - \frac{(v_0 + v_{ij})\cdot v_0}{\|v_0 + v_{ij}\|\|v_0\|}\right]  \label{eq:cut_lem_line_8}\\
&= \frac{\alpha_{GW}}{3}( 2 +  \gamma_{ij}).\nonumber
\end{align}
We get the equation (\ref{eq:cut_lem_line_2}) from the construction of the algorithm (the line 4 of Algorithm \ref{alg:alg}),
the equation (\ref{eq:cut_lem_line_3}) from averaging over the random vector $r$ over the sphere $S^{N-1}$ as in the original Goemans-Williamson paper \cite{GW95},
and finally the equation (\ref{eq:cut_lem_line_8}) from the equation (\ref{eq:sdp_property_2}) of Lemma \ref{lem:properties_of_sdp_sol}.
\end{proof}

\subsection{Proof of Lemma \ref{lem:energy_on_edge}}\label{appendix:energy_on_edge}
\begin{proof}
We largely take the analysis from \cite{AGM20} except that we allow different parameters on different edges. For the case of $z_i\ne z_j$, we will show something stronger; namely,

\begin{align*}
    \langle\phi(\vec\theta)|4H_{ij}|\phi(\vec\theta)\rangle 
    &= 1 + 
    \sin(2\theta_{ij}) \left[ \prod_{k\in N(i)\setminus \{j\}} \cos (2\theta_{ik}) 
    + \prod_{k\in N(j) \setminus \{i\}} \cos (2\theta_{kj})\right] \\
&+\sum_{\substack{\Delta'\subset \Delta_{ij} \\ |\Delta'|:\text{ even}}} \prod_{k\in \Delta'} \frac{\sin(2\theta_{ik})\sin(2\theta_{kj})}{\cos(2\theta_{ik}) \cos(2\theta_{kj}
)} \prod_{k\in N(i)\setminus \{j\}} \cos (2\theta_{ik}) 
     \prod_{k\in N(j) \setminus \{i\}} \cos (2\theta_{kj}),
\end{align*}
where $\Delta_{ij}:= N(i)\cap N(j).$ 
The inequality in the lemma is obtained by running the first summation just on $\Delta'=\phi$.

By symmetry, we can assume $z_i =0, z_j =1$. Then by construction, we have $P(i)P(j) = Y_i X_j$.
From the decomposition $4H_{ij} = 1 -X_i X_j - Y_i Y_j - Z_i Z_j$, we calculate the energy contribution $\langle\phi(\vec\theta)|4H_{ij}|\phi(\vec\theta)\rangle$ term by term. For the $X_i X_j$ term,
we have
\begin{align*}
&\langle\phi(\vec\theta)|X_i X_j|\phi(\vec\theta)\rangle \\
&=
\langle z|\prod_{(k,l)\in E} \exp[- i \theta_{kl} P(k) P(l)]X_i X_j \prod_{(k,l)\in E} \exp[i \theta_{kl} P(k) P(l)]|z\rangle \\
&=
\langle z|\prod_{k \in N(i)} \exp[- i 2\theta_{ik} P(i) P(k)]X_i X_j |z\rangle \\
&=
\langle z|\prod_{k \in N(i)} \left[\cos (2\theta_{ik}) - i\sin(2\theta_{ik})Y_i P(k) \right]X_i X_j |z\rangle \\
&=
-\langle z|i \sin(2\theta_{ij})Y_i X_j\prod_{k \in N(i)\setminus \{j\}} \cos (2\theta_{ik})  X_i X_j |z\rangle \\
&=
-\sin (2\theta_{ij}) \prod_{k \in N(i)\setminus \{j\}} \cos (2\theta_{ik}).
\end{align*}
Here, we used the standard commutator relations of Pauli matrices. From the line 4 and 5, we used the fact that the Pauli matrices on each qubit in $N(i)\setminus\{j\}$ should multiply to $I$ or $Z$ to yield a non-zero contribution.

Similar calculations yield 
\begin{align*}
    \langle\phi(\vec\theta) | Y_i Y_j |\phi(\vec\theta)\rangle = -\sin (2\theta_{ij}) \prod_{k \in N(j)\setminus \{i\}} \cos (2\theta_{kj}).
\end{align*}

For the $Z_i Z_j$ term,
\begin{align}
&\langle\phi(\vec\theta)|Z_i Z_j|\phi(\vec\theta)\rangle \nonumber\\
&=
\langle z|\prod_{(k,l)\in E} \exp[- i \theta_{kl} P(k) P(l)]Z_i Z_j \prod_{(k,l)\in E} \exp[i \theta_{kl} P(k) P(l)]|z\rangle \nonumber\\
&=
\langle z|\prod_{k \in N(i)\setminus\{j\}} \exp[- i 2\theta_{ik} P(i) P(k)] \prod_{k \in N(j)\setminus\{i\}} \exp[-i 2\theta_{kj} P(k) P(j)]Z_i Z_j|z\rangle  \label{eq:zz_line_3}\\
&=
\langle z|\prod_{k \in N(i)\setminus\{j\}} \big[\cos (2\theta_{ik}) - i \sin (2\theta_{ik})P(i)P(k)\big]\nonumber \\
&\quad \quad \quad \quad \quad \quad \quad \quad \quad \quad \quad
\prod_{k \in N(j)\setminus\{i\}} \big[\cos (2\theta_{kj}) - i \sin (2\theta_{kj})P(k)P(j)\big] Z_i Z_j|z\rangle \nonumber \\
&= 
\langle z|\prod_{k \in \Delta_{ij}} \big[\cos (2\theta_{ik}) - i \sin (2\theta_{ik})P(i)P(k)\big] \big[\cos (2\theta_{kj}) - i \sin (2\theta_{kj})P(k)P(j)\big] \label{eq:zz_line_5} \nonumber\\
&\quad \quad \quad \quad \quad \quad \quad \quad \quad \quad \quad
\prod_{k \in N(i)\setminus N(j)\setminus\{j\}}\cos(2\theta_{ik})  \prod_{k \in N(j)\setminus N(i)\setminus\{i\}}\cos(2\theta_{kj}) Z_i Z_j|z\rangle \\
&= 
\langle z|\prod_{k\in \Delta_{ij}} \big[\cos (2\theta_{ik}) \cos (2\theta_{kj}) -  \sin (2\theta_{ik})\sin (2\theta_{kj})P(i)P(j)\big]  \nonumber\\
&\quad \quad \quad \quad \quad \quad \quad \quad \quad \quad \quad
\prod_{k \in N(i)\setminus N(j)\setminus\{j\}}\cos(2\theta_{ik})  \prod_{k \in N(j)\setminus N(i)\setminus\{i\}}\cos(2\theta_{kj}) Z_i Z_j|z\rangle \nonumber\\
&=
\langle z|\prod_{k \in N(i)\cap N(j)} \left[1 -  \frac{\sin (2\theta_{ik}) \sin (2\theta_{kj})}{ \cos (2\theta_{ik}) \cos (2\theta_{kj})}Y_i  X_j\right] \prod_{k\in N(i)\setminus \{j\}} cos (2\theta_{ik}) \prod_{k\in N(j) \setminus \{i\}} cos (2\theta_{kj})Z_i Z_j |z\rangle \nonumber\\
&= 
- \sum_{\substack{\Delta'\subset \Delta_{ij} \\ |\Delta'|:\text{ even}}} \prod_{k\in \Delta'} \frac{\sin(2\theta_{ik})\sin(2\theta_{kj})}{\cos(2\theta_{ik}) \cos(2\theta_{kj}
)} \prod_{k\in N(i)\setminus \{j\}} \cos (2\theta_{ik}) 
     \prod_{k\in N(j) \setminus \{i\}} \cos (2\theta_{kj}) \label{eq:zz_line_last}
\end{align}

To get the equation (\ref{eq:zz_line_3}), we used the fact that only the Pauli matrices with non-trivial commutator with $Z_i Z_j$ survive. 
The equation (\ref{eq:zz_line_5}) and (\ref{eq:zz_line_last}) are from the fact that an even number of $P(i)$ should be applied on each qubit $i$ to give a non-zero contribution, because the state $|z\rangle$ is an eigenstate of the matrix $Z_i Z_j$.

In the case of $z_i = z_j$, the inequality (\ref{eq:energy_ij}) is 
trivial, because the Hamiltonian $H_{ij}$ is a PSD matrix.
In fact, 
we can express the energy as a function of the parameters and the connectivity on the graph as in the earlier case, but we choose not to show since the calculation is similar and is not necessary for the main theorem.
\end{proof}

\bibliography{ref}

\begin{thebibliography}{10}

\bibitem{AGM20}
Anurag Anshu, David Gosset, and Karen Morenz.
\newblock {Beyond Product State Approximations for a Quantum Analogue of Max
  Cut}.
\newblock In Steven~T. Flammia, editor, {\em 15th Conference on the Theory of
  Quantum Computation, Communication and Cryptography (TQC 2020)}, volume 158
  of {\em Leibniz International Proceedings in Informatics (LIPIcs)}, pages
  7:1--7:15, Dagstuhl, Germany, 2020. Schloss Dagstuhl--Leibniz-Zentrum f{\"u}r
  Informatik.
\newblock URL: \url{https://drops.dagstuhl.de/opus/volltexte/2020/12066}, \href
  {https://doi.org/10.4230/LIPIcs.TQC.2020.7}
  {\path{doi:10.4230/LIPIcs.TQC.2020.7}}.

\bibitem{AGMS21}
Anurag Anshu, David Gosset, Karen~J. Morenz~Korol, and Mehdi Soleimanifar.
\newblock Improved approximation algorithms for bounded-degree local
  hamiltonians.
\newblock {\em Phys. Rev. Lett.}, 127:250502, Dec 2021.
\newblock URL: \url{https://link.aps.org/doi/10.1103/PhysRevLett.127.250502},
  \href {https://doi.org/10.1103/PhysRevLett.127.250502}
  {\path{doi:10.1103/PhysRevLett.127.250502}}.

\bibitem{BBT07}
Nikhil Bansal, Sergey Bravyi, and Barbara~M. Terhal.
\newblock {C}lassical approximation schemes for the ground-state energy of
  quantum and classical {I}sing spin {H}amiltonians on planar graphs.
\newblock Quant. Inf. Comp. Vol. 9, No.8, p. 0701 (2009), 2007.
\newblock arXiv:0705.1115v4.
\newblock \href {http://arxiv.org/abs/0705.1115} {\path{arXiv:0705.1115}}.

\bibitem{BH16}
Fernando G. S.~L. Brand{\~a}o and Aram~W. Harrow.
\newblock Product-state approximations to quantum states.
\newblock {\em Communications in Mathematical Physics}, 342(1):47--80, Feb
  2016.
\newblock \href {https://doi.org/10.1007/s00220-016-2575-1}
  {\path{doi:10.1007/s00220-016-2575-1}}.

\bibitem{BGKT19}
Sergey Bravyi, David Gosset, Robert Ko\"nig, and Kristan Temme.
\newblock Approximation algorithms for quantum many-body problems.
\newblock {\em Journal of Mathematical Physics}, 60(3):032203, 2019.
\newblock \href {http://arxiv.org/abs/https://doi.org/10.1063/1.5085428}
  {\path{arXiv:https://doi.org/10.1063/1.5085428}}, \href
  {https://doi.org/10.1063/1.5085428} {\path{doi:10.1063/1.5085428}}.

\bibitem{VQA21}
Marco Cerezo, Andrew Arrasmith, Ryan Babbush, Simon~C Benjamin, Suguru Endo,
  Keisuke Fujii, Jarrod~R McClean, Kosuke Mitarai, Xiao Yuan, Lukasz Cincio,
  et~al.
\newblock Variational quantum algorithms.
\newblock {\em Nature Reviews Physics}, 3(9):625--644, 2021.

\bibitem{FGG14}
Edward Farhi, Jeffrey Goldstone, and Sam Gutmann.
\newblock A quantum approximate optimization algorithm, 2014.
\newblock URL: \url{https://arxiv.org/abs/1411.4028}, \href
  {https://doi.org/10.48550/ARXIV.1411.4028}
  {\path{doi:10.48550/ARXIV.1411.4028}}.

\bibitem{GK11}
Sevag Gharibian and Julia Kempe.
\newblock {A}pproximation algorithms for {QMA}-complete problems.
\newblock SIAM Journal on Computing 41(4): 1028-1050, 2012, 2011.
\newblock arXiv:1101.3884v1.
\newblock \href {http://arxiv.org/abs/1101.3884} {\path{arXiv:1101.3884}},
  \href {https://doi.org/10.1137/110842272} {\path{doi:10.1137/110842272}}.

\bibitem{GK12}
Sevag Gharibian and Julia Kempe.
\newblock {\em Hardness of approximation for quantum problems}, pages 387--398.
\newblock Springer Berlin Heidelberg, Berlin, Heidelberg, 2012.
\newblock URL: \url{http://dx.doi.org/10.1007/978-3-642-31594-7_33}, \href
  {https://doi.org/10.1007/978-3-642-31594-7_33}
  {\path{doi:10.1007/978-3-642-31594-7_33}}.

\bibitem{GP19}
Sevag Gharibian and Ojas Parekh.
\newblock {Almost Optimal Classical Approximation Algorithms for a Quantum
  Generalization of Max-Cut}.
\newblock In Dimitris Achlioptas and L{\'a}szl{\'o}~A. V{\'e}gh, editors, {\em
  Approximation, Randomization, and Combinatorial Optimization. Algorithms and
  Techniques (APPROX/RANDOM 2019)}, volume 145 of {\em Leibniz International
  Proceedings in Informatics (LIPIcs)}, pages 31:1--31:17, Dagstuhl, Germany,
  2019. Schloss Dagstuhl--Leibniz-Zentrum fuer Informatik.
\newblock URL: \url{http://drops.dagstuhl.de/opus/volltexte/2019/11246}, \href
  {https://doi.org/10.4230/LIPIcs.APPROX-RANDOM.2019.31}
  {\path{doi:10.4230/LIPIcs.APPROX-RANDOM.2019.31}}.

\bibitem{GW95}
Michel~X. Goemans and David~P. Williamson.
\newblock Improved approximation algorithms for maximum cut and satisfiability
  problems using semidefinite programming.
\newblock {\em J. ACM}, 42(6):1115--1145, November 1995.
\newblock URL: \url{http://doi.acm.org/10.1145/227683.227684}, \href
  {https://doi.org/10.1145/227683.227684} {\path{doi:10.1145/227683.227684}}.

\bibitem{HLP20}
Sean Hallgren, Eunou Lee, and Ojas Parekh.
\newblock An approximation algorithm for the max-2-local hamiltonian problem.
\newblock In Jaroslaw Byrka and Raghu Meka, editors, {\em Approximation,
  Randomization, and Combinatorial Optimization. Algorithms and Techniques,
  {APPROX/RANDOM} 2020, August 17-19, 2020, Virtual Conference}, volume 176 of
  {\em LIPIcs}, pages 59:1--59:18. Schloss Dagstuhl - Leibniz-Zentrum f{\"{u}}r
  Informatik, 2020.
\newblock \href {https://doi.org/10.4230/LIPIcs.APPROX/RANDOM.2020.59}
  {\path{doi:10.4230/LIPIcs.APPROX/RANDOM.2020.59}}.

\bibitem{HNPTW21}
Yeongwoo Hwang, Joe Neeman, Ojas Parekh, Kevin Thompson, and John Wright.
\newblock Unique games hardness of quantum max-cut, and a vector-valued
  borell's inequality, 2021.
\newblock URL: \url{https://arxiv.org/abs/2111.01254}, \href
  {https://doi.org/10.48550/ARXIV.2111.01254}
  {\path{doi:10.48550/ARXIV.2111.01254}}.

\bibitem{KP22}
John Kallaugher and Ojas Parekh.
\newblock The quantum and classical streaming complexity of quantum and
  classical max-cut.
\newblock {\em CoRR}, abs/2206.00213, 2022.
\newblock \href {http://arxiv.org/abs/2206.00213} {\path{arXiv:2206.00213}},
  \href {https://doi.org/10.48550/arXiv.2206.00213}
  {\path{doi:10.48550/arXiv.2206.00213}}.

\bibitem{KKEO07}
Subhash Khot, Guy Kindler, Elchanan Mossel, and Ryan O’Donnell.
\newblock Optimal inapproximability results for max‐cut and other
  2‐variable csps?
\newblock {\em SIAM Journal on Computing}, 37(1):319--357, 2007.
\newblock \href
  {http://arxiv.org/abs/https://doi.org/10.1137/S0097539705447372}
  {\path{arXiv:https://doi.org/10.1137/S0097539705447372}}, \href
  {https://doi.org/10.1137/S0097539705447372}
  {\path{doi:10.1137/S0097539705447372}}.

\bibitem{KCM21}
Karen J~Morenz Korol, Kenny Choo, and Antonio Mezzacapo.
\newblock Quantum approximation algorithms for many-body and electronic
  structure problems.
\newblock {\em arXiv preprint arXiv:2111.08090}, 2021.

\bibitem{ANP08}
Miguel Navascues, Stefano Pironio, and Antonio Acín.
\newblock Convergent hierarchy of semidefinite programs characterizing the set
  of quantum correlations.
\newblock {\em New Journal of Physics}, 10, 07 2008.
\newblock \href {https://doi.org/10.1088/1367-2630/10/7/073013}
  {\path{doi:10.1088/1367-2630/10/7/073013}}.

\bibitem{PT21a}
Ojas Parekh and Kevin Thompson.
\newblock {Application of the Level-2 Quantum Lasserre Hierarchy in Quantum
  Approximation Algorithms}.
\newblock In Nikhil Bansal, Emanuela Merelli, and James Worrell, editors, {\em
  48th International Colloquium on Automata, Languages, and Programming (ICALP
  2021)}, volume 198 of {\em Leibniz International Proceedings in Informatics
  (LIPIcs)}, pages 102:1--102:20, Dagstuhl, Germany, 2021. Schloss Dagstuhl --
  Leibniz-Zentrum f{\"u}r Informatik.
\newblock URL: \url{https://drops.dagstuhl.de/opus/volltexte/2021/14171}, \href
  {https://doi.org/10.4230/LIPIcs.ICALP.2021.102}
  {\path{doi:10.4230/LIPIcs.ICALP.2021.102}}.

\bibitem{PT21b}
Ojas Parekh and Kevin Thompson.
\newblock Beating random assignment for approximating quantum 2-local
  hamiltonian problems.
\newblock In Petra Mutzel, Rasmus Pagh, and Grzegorz Herman, editors, {\em 29th
  Annual European Symposium on Algorithms, {ESA} 2021, September 6-8, 2021,
  Lisbon, Portugal (Virtual Conference)}, volume 204 of {\em LIPIcs}, pages
  74:1--74:18. Schloss Dagstuhl - Leibniz-Zentrum f{\"{u}}r Informatik, 2021.
\newblock \href {https://doi.org/10.4230/LIPIcs.ESA.2021.74}
  {\path{doi:10.4230/LIPIcs.ESA.2021.74}}.

\bibitem{PT22}
Ojas Parekh and Kevin Thompson.
\newblock An optimal product-state approximation for 2-local quantum
  hamiltonians with positive terms.
\newblock {\em CoRR}, abs/2206.08342, 2022.
\newblock \href {http://arxiv.org/abs/2206.08342} {\path{arXiv:2206.08342}},
  \href {https://doi.org/10.48550/arXiv.2206.08342}
  {\path{doi:10.48550/arXiv.2206.08342}}.

\bibitem{PM17}
Stephen Piddock and Ashley Montanaro.
\newblock The complexity of antiferromagnetic interactions and 2d lattices.
\newblock {\em Quantum Inf. Comput.}, 17(7{\&}8):636--672, 2017.
\newblock \href {https://doi.org/10.26421/QIC17.7-8-6}
  {\path{doi:10.26421/QIC17.7-8-6}}.

\end{thebibliography}

\end{document}